\documentclass[11pt,a4paper]{article}
\usepackage[utf8]{inputenc}
\usepackage{units}
\usepackage{amsmath}
\usepackage{amsthm}
\usepackage{stmaryrd}
\usepackage{authblk}
\usepackage[margin=3cm]{geometry}

\makeatletter

\theoremstyle{plain}
\newtheorem{thm}{\protect Theorem}
\theoremstyle{definition}
\newtheorem{defn}{\protect Definition}

\title{Quantum Mechanics is Compatible with Counterfactual Definiteness}

\author[1]{Janne V. Kujala}
\author[2]{Ehtibar {N.}~Dzhafarov}

\affil[1]{Department of Mathematics and Statistics, 
 University of Turku, Finland; jvk@iki.fi}
\affil[2]{Department of Psychological Sciences, Purdue University, USA; ehtibar@purdue.edu}

\date{September 21, 2023}

\makeatother

\begin{document}
\maketitle

\abstract{Counterfactual definiteness (CFD) means that if some property is
measured in some context, then the outcome of the measurement would
have been the same had this property been measured in a different
context. A context includes all other measurements made together with
the one in question, and the spatiotemporal relations among them.
The proviso for CFD is non-disturbance: any physical influence of
the contexts on the property being measured is excluded by the laws
of nature, so that no one measuring this property has a way of ascertaining
its context. It is usually claimed that in quantum mechanics CFD does
not hold, because if one assigns the same value to a property in all
contexts it is measured in, one runs into a logical contradiction,
or at least contravenes quantum theory and experimental evidence.
We show that this claim is not substantiated if one takes into account
that only one of the possible contexts can be a factual context, all
other contexts being counterfactual. With this in mind, any system
of random variables can be viewed as satisfying CFD. The concept of
CFD is closely related to but distinct from that of noncontextuality,
and it is the latter property that may or may not hold for a system,
in particular being contravened by some quantum systems.

}

\section{Introduction}

A measurement has three characteristics. One is the measurement's
\emph{content}: this is the question the measurement answers, or equivalently,
the physical property whose value the measurement determines. The
second characteristic of a measurement is its \emph{context}: this
includes other measurements made together with this one, and the spatiotemporal
relations among them. The ``togetherness'' of two measurements means
that there is an empirical rule by which the outcomes of these measurements
are paired. The third characteristic of a measurement is the \emph{probability
distributio}n of its values. More precisely, all the measurements
made in the same context possess a \emph{joint distribution} which
determines the distribution of any given measurement.

\emph{Counterfactual definiteness} (CFD) of a measurement is its compliance
with the following counterfactual statement: had the measurement with
the same content been made in another context, its outcome would have
been the same. We will argue that under an assumption commonly accepted
in quantum mechanics, CFD is always satisfied.

We begin by illustrating the terms and notions mentioned above (and
to be formally defined in Section~\ref{sec:Formal-treatment}) using
a toy example. It is based on the parable of the seer of Nineveh that
was introduced by Ernst Specker \cite{Specker1960} and subsequently
used by others as a simple example of contextuality \cite{Liangetal.2011}.
Omitting the colorful story line, the seer of Nineveh had three boxes
obeying the following Magic Box Rules: (MBR1) only two of them could
be opened at any given time; (MBR2) regardless of which two boxes
were opened, one and only one of them contained a gem, and (MBR3)
the gem could be contained in either of the two with equal probabilities.
We can take these Magic Box Rules as an analogue of the laws of quantum
mechanics. A formal representation of this situation is by the following
\emph{system of random variables}:
\begin{equation}
\begin{array}{|c|c|c||c|}
\hline R_{1}^{1} & R_{2}^{1} &  & c=1\\
\hline  & R_{2}^{2} & R_{3}^{2} & c=2\\
\hline R_{1}^{3} &  & R_{3}^{3} & c=3\\
\hline\hline q=1 & q=2 & q=3 & \textnormal{system }\mathcal{R}_{0}
\\\hline \end{array}\label{eq:magic}
\end{equation}
It represents six measurements $R_{q}^{c}$ having three \emph{contents}
$q$ and made pairwise in three \emph{contexts} $c$. The content
$q$ of $R_{q}^{c}$ can be thought of as the question ``does the
box $q$ contain a gem?'' ($q=1,2,3$). This question is answered
$Yes$ or $No$, which are the possible values of $R_{q}^{c}$. Equivalently,
we can view $q$ as the property ``the contents of the box $q$''
($q=1,2,3$), in which case the possible values of $R_{q}^{c}$ are
``gem'' and ``no gem.'' Irrespective, we will denote the values
of $R_{q}^{c}$ as $+1$ and $-1$. The context $c$ indicates which
two boxes have been opened. All other conditions under which the measurement
are made (e.g., the shape of the boxes) are the same in the three
contexts, so we do not list them in the definition of a context. Contexts
are always mutually exclusive, by definition: random variables measured
in different conditions never co-occur, there is no empirical rule
for pairing the values of, say, $R_{1}^{1}$ and $R_{3}^{2}$.

The rules MBR2 and MBR3 say that in each of the three contexts, the
two measurements $R_{q}^{c}$ and $R_{q'}^{c}$ made in it have the
joint distribution 
\begin{equation}
\begin{array}{|c|c|c||c|}
\hline  & R_{q}^{c}=+1 & R_{q}^{c}=-1 & \\
\hline R_{q'}^{c}=+1 & 0 & \nicefrac{1}{2} & \nicefrac{1}{2}\\
\hline R_{q'}^{c}=-1 & \nicefrac{1}{2} & 0 & \nicefrac{1}{2}\\
\hline\hline  & \nicefrac{1}{2} & \nicefrac{1}{2} & 
\\\hline \end{array}
\end{equation}
The individual distribution of each $R_{q}^{c}$, by MBR3, is one
and the same for all random variables: $+1$ and $-1$ with equal
probabilities. In particular, it is the same for any two measurements
answering the same question in different contexts, such as $R_{2}^{1}$
and $R_{2}^{2}$. If one repeatedly observes openings of the box $q=2$,
and sees no other boxes, one has no way of determining in which context
the box is being opened, in the one with the box $q=1$ or with the
box $q=3$. This is a special case of an assumption we are going to
make throughout this paper: in the quantum mechanical literature it
is known under variety of names, such as \emph{non-signaling} or \emph{non-disturbance}
\cite{Ramanathanetal.2012,Cereceda,PR1994}. In Section~\ref{sec:Formal-treatment},
we will define a strong version of this notion following Abramsky
and Brandenburger \cite{AbramskyBrandenburger(2011)}.

\subsection{Noncontextual Representation of Variables }

A standard way of introducing the notion of contextuality, applying
it to our example, is to ask: 
\begin{enumerate}
\item[{Q0:}] is it possible to treat all random variables in the system
as if any two variables with the same content were identical? 
\end{enumerate}
For instance, $R_{1}^{1}$ and $R_{1}^{3}$ have the same content,
both answer the question $q=1$: ``Does the box \#1 contain a gem?''.
The other boxes opened together with the box \#1 (i.e., the box \#2
in context $c=1$ or the box \#3 in context $c=3$) in no way affect
the distribution of the possible answers to the question $q=1$. Someone
who observes the box \#1 repeatedly, without seeing the other boxes,
has no way of determining the context of the box \#1 when it is opened.
Therefore it seems it should be possible to simply view $R_{1}^{1}$
and $R_{1}^{3}$ as one and the same variable. Analogous reasoning
applies to other pairs of measurements sharing a content, $\left(R_{2}^{1},R_{2}^{2}\right)$
and $\left(R_{3}^{2},R_{3}^{3}\right)$. 

However, it is easy to see that this \emph{noncontextual representation}
of the variables in our example is not possible. Let us begin by renaming
$R_{1}^{1}$ into $X$, and then proceed by identifying other random
variables following the Magic Box Rules and noncontextual representation.
The first step will yield 
\begin{equation}
\begin{array}{|c|c|c||c|}
\hline X &  &  & c=1\\
\hline  &  &  & c=2\\
\hline  &  &  & c=3\\
\hline\hline q=1 & q=2 & q=3 & \mathcal{R}_{0}
\\\hline \end{array}\Rightarrow\begin{array}{|c|c|c||c|}
\hline X & -X &  & c=1\\
\hline  &  &  & c=2\\
\hline X &  &  & c=3\\
\hline\hline q=1 & q=2 & q=3 & \mathcal{R}_{0}
\\\hline \end{array}
\end{equation}
where the $-X$ in the first row follows MBR2, and the second $X$
in the first column follows noncontextual representability. Proceeding
in this manner, we obtain
\begin{equation}
\begin{array}{ccc}
\begin{array}{|c|c|c||c|}
\hline X & -X &  & c=1\\
\hline  &  &  & c=2\\
\hline X &  &  & c=3\\
\hline\hline q=1 & q=2 & q=3 & \mathcal{R}_{0}
\\\hline \end{array} & \Rightarrow & \begin{array}{|c|c|c||c|}
\hline X & -X &  & c=1\\
\hline  & -X &  & c=2\\
\hline X &  & -X & c=3\\
\hline\hline q=1 & q=2 & q=3 & \mathcal{R}_{0}
\\\hline \end{array}\\
 &  & \Downarrow\\
 &  & \begin{array}{|c|c|c||c|}
\hline X & -X &  & c=1\\
\hline  & -X & ? & c=2\\
\hline X &  & -X & c=3\\
\hline\hline q=1 & q=2 & q=3 & \mathcal{R}_{0}
\\\hline \end{array}
\end{array}
\end{equation}
and we see that the cell with ``?'' cannot be filled, as it should
be $-X$ to maintain noncontextual representation in the third column
but it should be $X$ to follow MBR2 in the second row. The conclusion
is that no noncontextual representation of the random variables in
our system exists. When this happens, a system is said to be\emph{
contextual} (otherwise it is \emph{noncontextual}). 

\subsection{Counterfactual Definiteness}

One can, however, approach our system in a different way. Given that
a box was opened in some context (that we will call the \emph{factual
context}), one can ask: had this box been opened in another context
(called \emph{counterfactual}), would the outcome have been the same?
The term ``outcome'' has two meanings: ``random variable'' and
``value of a random variable.'' In the present context, however,
the two are interchangeable, and the counterfactual question can also
be formulated thus: had this box been opened in another context, would
the counterfactual variable $R'$ have been representable by the same
random variable as the factual one, $R$? The reason for this is that
we can think of the counterfactual question about values of the variables
being asked repeatedly, and $R'$ and $R$ can always have the same
value if and only if $R'=R$.

In the contextuality literature the counterfactual question above
is considered to be logically equivalent to Q0 \cite{Dzh2019,Peres1978,Stapp1998,Svozil2009,Stapp1971,CS1981}.
However, a detailed analysis shows this is not the case. Using our
example (\ref{eq:magic}), consider the situation when the factual
context is $c=2$, i.e.\ we observe the values of the \emph{factual
variables} $R_{2}^{2}$ and $R_{3}^{2}$. One can then ask two counterfactual
questions: 
\begin{enumerate}
\item[{Q1:}] if instead of $R_{2}^{2}$ one had recorded $R_{2}^{1}$
(the same box in context $c=1$), would $R_{2}^{1}$ have been the
same as $R_{2}^{2}$? 
\item[{Q2:}] if instead of $R_{3}^{2}$ one had recorded $R_{3}^{3}$
(the same box in context $c=3$), would $R_{3}^{3}$ have been the
same as $R_{3}^{2}$? 
\end{enumerate}
It is easy to see that in our example the answer to both questions
is affirmative, in the sense that there is nothing in the Magic Box
Rules that would prevent one from considering a counterfactual variable
identical to the corresponding factual one. Denoting $R_{2}^{2}$
by $X$, we have, for Q1,

\begin{equation}
\begin{array}{|c|c|c||c|}
\hline  &  &  & c=1\\
\hline  & X & -X & \boxed{c=2}\vphantom{\begin{array}{c}
\\
\\
\end{array}}\\
\hline  &  &  & c=3\\
\hline\hline q=1 & q=2 & q=3 & \mathcal{R}_{0}
\\\hline \end{array}\Rightarrow\begin{array}{|c|c|c||c|}
\hline  & X &  & c=1\\
\hline  & X & -X & \boxed{c=2}\vphantom{\begin{array}{c}
\\
\\
\end{array}}\\
\hline  &  &  & c=3\\
\hline\hline q=1 & q=2 & q=3 & \mathcal{R}_{0}
\\\hline \end{array}
\end{equation}
Moreover, this representation can never come into a conflict with
other variables in the same counterfactual context: 
\begin{equation}
\begin{array}{|c|c|c||c|}
\hline  & X &  & c=1\\
\hline  & X & -X & \boxed{c=2}\vphantom{\begin{array}{c}
\\
\\
\end{array}}\\
\hline  &  &  & c=3\\
\hline\hline q=1 & q=2 & q=3 & \mathcal{R}_{0}
\\\hline \end{array}\Rightarrow\begin{array}{|c|c|c||c|}
\hline -X & X &  & c=1\\
\hline  & X & -X & \boxed{c=2}\vphantom{\begin{array}{c}
\\
\\
\end{array}}\\
\hline  &  &  & c=3\\
\hline\hline q=1 & q=2 & q=3 & \mathcal{R}_{0}
\\\hline \end{array}
\end{equation}
We can repeat the same reasoning for Q2: 
\begin{equation}
\begin{array}{ccc}
\begin{array}{|c|c|c||c|}
\hline  &  &  & c=1\\
\hline  & X & -X & \boxed{c=2}\vphantom{\begin{array}{c}
\\
\\
\end{array}}\\
\hline  &  &  & c=3\\
\hline\hline q=1 & q=2 & q=3 & \mathcal{R}_{0}
\\\hline \end{array} & \Rightarrow & \begin{array}{|c|c|c||c|}
\hline  &  &  & c=1\\
\hline  & X & -X & \boxed{c=2}\vphantom{\begin{array}{c}
\\
\\
\end{array}}\\
\hline  &  & -X & c=3\\
\hline\hline q=1 & q=2 & q=3 & \mathcal{R}_{0}
\\\hline \end{array}\\
 &  & \Downarrow\\
 &  & \begin{array}{|c|c|c||c|}
\hline  &  &  & c=1\\
\hline  & X & -X & \boxed{c=2}\vphantom{\begin{array}{c}
\\
\\
\end{array}}\\
\hline X &  & -X & c=3\\
\hline\hline q=1 & q=2 & q=3 & \mathcal{R}_{0}
\\\hline \end{array}
\end{array}
\end{equation}
We have here a special case of the general theorem proved in the next
section: it says that if a system satisfies the no-disturbance condition,
then for any factual context and any counterfactual one, the variables
in the latter can be chosen so that the same-content variables in
the two contexts are identical. That is, any system with no disturbance
has the property of CFD.

Returning to our example, however, we have a natural question to ask:
What if the questions Q1 and Q2 are answered together? Would we not run
into a contradiction then? What we have is
\begin{equation}
\begin{array}{ccc}
\begin{array}{|c|c|c||c|}
\hline  &  &  & c=1\\
\hline  & X & -X & \boxed{c=2}\vphantom{\begin{array}{c}
\\
\\
\end{array}}\\
\hline  &  &  & c=3\\
\hline\hline q=1 & q=2 & q=3 & \mathcal{R}_{0}
\\\hline \end{array} & \Rightarrow & \begin{array}{|c|c|c||c|}
\hline  & X &  & c=1\\
\hline  & X & -X & \boxed{c=2}\vphantom{\begin{array}{c}
\\
\\
\end{array}}\\
\hline  &  & -X & c=3\\
\hline\hline q=1 & q=2 & q=3 & \mathcal{R}_{0}
\\\hline \end{array}\\
 &  & \Downarrow\\
 &  & \begin{array}{|c|c|c||c|}
\hline -X & X &  & c=1\\
\hline  & X & -X & \boxed{c=2}\vphantom{\begin{array}{c}
\\
\\
\end{array}}\\
\hline X &  & -X & c=3\\
\hline\hline q=1 & q=2 & q=3 & \mathcal{R}_{0}
\\\hline \end{array}
\end{array}\label{eq:Q1Q2}
\end{equation}
And it seems that we indeed have run into a contradiction, because
in the first column the variables are not the same. However, one can
notice this only if one compares two counterfactual contexts to each
other with the purpose of determining if they comply with noncontextual
representability. In other words, one notices this contradiction if
the question one answers is Q0 rather than Q1 and Q2. 

We already know that the system is contextual, i.e., Q0 is answered
in the negative. What we should be interested in now is whether a
contradiction occurs if we deal only with the counterfactual questions,
without explicitly involving noncontextual representability. It is
clear, however, that one cannot formulate purely counterfactual questions
to compare two counterfactual contexts without making one of them
factual. It is logically impossible.

The principal difference between noncontextual representability and
CFD is that the latter puts the system into a frame of reference formed
by the choice of a factual context. Changing the factual context changes
the frame of reference. In this picture, noncontextual representability
can be viewed as the possibility of reconciling all different frames
of reference. However, this is an additional and different question
--- about contextuality. A system may be contextual or noncontextual,
but CFD is satisfied always.

\section{\label{sec:Formal-treatment}Formal Treatment of Contextuality and
Counterfactual Definiteness}

\subsection{Basic Notions}

We begin by defining the notions discussed in the previous section
in a more rigorous way. The terminology and notation we use are those
developed in the Contextuality-by-Default (CbD) approach to contextuality
\cite{Dzh2022,DK2016}. Although we have presented an example of
a system of random variables in the opening section, it was a very
specially constructed system (uniform dichotomous distributions and
perfect anticorrelations in each context). We think therefore it is
useful to provide additional illustrations using an example of a more
generic variety:
\begin{equation}
\begin{array}{|c|c|c|c|c||c|}
\hline R_{1}^{1} & R_{2}^{1} & R_{3}^{1} &  &  & c=1\\
\hline  & R_{2}^{2} & R_{3}^{2} & R_{4}^{2} &  & c=2\\
\hline R_{1}^{3} &  & R_{3}^{3} &  &  & c=3\\
\hline R_{1}^{4} & R_{2}^{4} & R_{3}^{4} &  & R_{5}^{4} & c=4\\
\hline\hline q=1 & 2 & 3 & 4 & 5 & \mathcal{R}_{1}
\\\hline \end{array}\label{eq:example system}
\end{equation}
In parallel, we will also use for illustrations a realisitic example,
the system of random variables for which John Bell and others derived
the celebrated inequalities bearing his name \cite{Bell1966,CHSH1969,Fine1982}:
\begin{equation}
\begin{array}{|c|c|c|c||c|}
\hline R_{1}^{1} & R_{2}^{1} &  &  & c=1\\
\hline  & R_{2}^{2} & R_{3}^{2} &  & c=2\\
\hline  &  & R_{3}^{3} & R_{4}^{3} & c=3\\
\hline R_{1}^{4} &  &  & R_{4}^{4} & c=4\\
\hline\hline q=1 & 2 & 3 & 4 & \mathcal{R}_{2}
\\\hline \end{array}\label{eq:example EPR/B}
\end{equation}
This system describes the EPR/Bohm experiment \cite{BohmAharonov1957}
with two entangled spin-$\nicefrac{1}{2}$ particles whose spins are
measured by two respective spacelike-separated experimenters traditionally
designated as Alice and Bob. The contents $q=1$ and $q=3$ designate
the settings (axes) that may be chosen by Alice, and Bob's settings
are designated by $q=2$ and $q=4$. Mathematically, system $\mathcal{R}_{2}$
is less interesting than system $\mathcal{R}_{1}$ (the former being
essentially of the same structure as our opening toy example). However,
$\mathcal{R}_{2}$ has the distinction of having dominated the discussions
related to contextuality (in the form of nonlocality) in the literature
on the foundations of quantum mechanics.

In complete generality, a system of random variables is an indexed
set 
\begin{equation}
\mathcal{R}=\left\{ R_{q}^{c}:q\in Q,\,c\in C,\,q\prec c\right\} ,\label{eq:system}
\end{equation}
where $Q$ and $C$ are sets of contents and contexts, respectively,
and $q\prec c$ indicates that $q$ is measured in $c$, with the
outcome $R_{q}^{c}$ a random variable. In each context the variables
possess a joint distribution, whereas there are no joint distributions
across the contexts.

The notion defined next uses a CbD term for what is usually referred
to as non-disturbance or non-signaling, understood in the strong sense
of the term formalized by Abramsky and Brandenburger \cite{AbramskyBrandenburger(2011)}.
\begin{defn}
A system $\mathcal{R}$ in (\ref{eq:system}) is \emph{strongly consistently
connected} (s.c.c.)\ if, for any $Q'\subseteq Q$ and any $c\in C$
such that $q\prec c$ for all $q\in Q'$, the joint probability distribution
of $\left\{ R_{q}^{c}:q\in Q'\right\} $ only depends on $Q'$ \cite{DK2016}.
\end{defn}
We will assume that both our example systems, (\ref{eq:example system})
and (\ref{eq:example EPR/B}), are s.c.c. (for system $\mathcal{R}_{2}$
this follows from the spacelike separation of Alice and Bob). In $\mathcal{R}_{1}$,
choosing $Q'$ as $\left\{ 2,3\right\} $, we have the identically
distributed pairs $\left\{ R_{2}^{1},R_{3}^{1}\right\} $, $\left\{ R_{2}^{2},R_{3}^{2}\right\} $,
and $\left\{ R_{2}^{4},R_{3}^{4}\right\} $. In $\mathcal{R}_{2}$,
choosing again $Q'$ as $\left\{ 2,3\right\} $, we have identically
distributed $R_{3}^{2}$ and $R_{3}^{3}$.
\begin{defn}
\label{def:s.c.c. noncontextual}An s.c.c.-system $\mathcal{R}$ in
(\ref{eq:system}) is \emph{noncontextual} if there is a jointly distributed
set of random variables $S=\left\{ S_{q}:q\in Q\right\} $, such that,
for any $Q'\subseteq Q$ and any $c\in C$ with $q\prec c$ for all
$q\in Q'$, the distribution of $\left\{ R_{q}^{c}:q\in Q'\right\} $
is the same as the distribution of $\left\{ S_{q}:q\in Q'\right\} $.
The set of variables $S$ is referred to as a \emph{reduced coupling}
of $\mathcal{R}$ \cite{DK2016Handbook}.
\end{defn}
Applying this definition to  system $\mathcal{R}_{1}$ in (\ref{eq:example system}),
it is noncontextual if one can find five jointly distributed variables
$\left\{ S_{1},\ldots,S_{5}\right\} $ such that 
\begin{equation}
\left\{ R_{1}^{1},R_{2}^{1},R_{3}^{1}\right\} \overset{d}{=}\left\{ S_{1},S_{2},S_{3}\right\} ,\quad\left\{ R_{2}^{2},R_{3}^{2},R_{4}^{2}\right\} \overset{d}{=}\left\{ S_{2},S_{3},S_{4}\right\} ,\quad\text{etc.},
\end{equation}
with $\overset{d}{=}$ standing for equality of distributions. For
system $\mathcal{R}_{2}$, noncontextuality means the existence of
four jointly distributed variables $\left\{ S'_{1},\ldots,S'_{4}\right\} $
such that 
\begin{equation}
\left\{ R_{1}^{1},R_{2}^{1}\right\} \overset{d}{=}\left\{ S_{1},S_{2}\right\} ,\quad\left\{ R_{2}^{2},R_{3}^{2}\right\} \overset{d}{=}\left\{ S_{2},S_{3}\right\} ,\quad\text{etc}.
\end{equation}
It is well known that this condition is equivalent to 
\begin{equation}
\max\left(\pm\left\langle R_{1}^{1}R_{2}^{1}\right\rangle \pm\left\langle R_{2}^{2}R_{3}^{2}\right\rangle \pm\left\langle R_{3}^{3}R_{4}^{3}\right\rangle \pm\left\langle R_{4}^{4}R_{1}^{4}\right\rangle \right)\leq2,\label{eq:Bell}
\end{equation}
where the maximum is taken over the eight choices of the $\pm$ signs
with odd number of minus signs. (This is the CHSH inequality \cite{CHSH1969},
with all variables' values assumed to be $\pm1$.)

The next notion formalizes the intuitive meaning of the following
statement: The essence of noncontextuality for s.c.c.\ systems is
that all content-sharing random variables can be treated \emph{as
if they were} one and the same variable. 
\begin{defn}
An \emph{identically connected (i.c.)}\ \emph{coupling} of a noncontextual
s.c.c.\ system $\mathcal{R}$ in (\ref{eq:system}) is an indexed
set
\begin{equation}
S^{*}=\left\{ S_{q}^{c}:S_{q}^{c}=S_{q},\,q\in Q,\,c\in C,\,q\prec c\right\} ,
\end{equation}
where $S_{q}$ is an element of the reduced coupling of the system.
\end{defn}
Thus, in our two examples, (\ref{eq:example system}) and (\ref{eq:example EPR/B}),
if the systems $\mathcal{R}_{1}$ and $\mathcal{R}_{2}$ are noncontextual,
then their variables can be viewed \emph{as if they were}, respectively,
\begin{equation}
\begin{array}{|c|c|c|c|c||c|}
\hline S_{1} & S_{2} & S_{3} &  &  & c=1\\
\hline  & S_{2} & S_{3} & S_{4} &  & c=2\\
\hline S_{1} &  & S_{3} &  &  & c=3\\
\hline S_{1} & S_{2} & S_{3} &  & S_{5} & c=4\\
\hline\hline q=1 & 2 & 3 & 4 & 5 & \widetilde{\mathcal{R}}_{1}\vphantom{\begin{array}{c}
\\
\\
\end{array}}
\\\hline \end{array}\mathrm{\:,\:\begin{array}{|c|c|c|c||c|}
\hline S'_{1} & S'_{2} &  &  & c=1\\
\hline  & S'_{2} & S'_{3} &  & c=2\\
\hline  &  & S'_{3} & S'_{4} & c=3\\
\hline S'_{1} &  &  & S'_{4} & c=4\\
\hline\hline q=1 & 2 & 3 & 4 & \widetilde{\mathcal{R}}_{2}\vphantom{\begin{array}{c}
\\
\\
\end{array}}
\\\hline \end{array}}
\end{equation}
These systems of variables are i.c.\ couplings of, respectively,
$\mathcal{R}_{1}$ and $\mathcal{R}_{2}$.

\subsection{Factual-Counterfactual Subsystems}

Any context in a system can be chosen and designated as a \emph{factual
context}. The variables recorded in this context are called \emph{factual
variables}. All other contexts and the variables they contain are
referred to as \emph{counterfactual}.
\begin{defn}
Having chosen a factual context, $c=c_{0}$, a subsystem $\mathcal{R}\left\llbracket c_{0}\right\rrbracket $
of the system $\mathcal{R}$ in (\ref{eq:system}) is called a \emph{factual-counterfactual}
(F-CF) subsystem (with respect to $c_{0}$) if it consists of all
variables that share their contents with the factual variables.
\end{defn}
This subsystem, of course, includes the factual variables themselves.
Presented explicitly, the F-CF subsystem of (\ref{eq:system}) with
respect to $c_{0}$ is
\begin{equation}
\mathcal{R}\left\llbracket c_{0}\right\rrbracket =\left\{ R_{q}^{c}:q\in Q,\,c\in C,\,q\prec c,c_{0}\right\} ,
\end{equation}
where $q\prec c,c_{0}$ means $q\prec c$ and $q\prec c_{0}$. Thus,
in the examples (\ref{eq:example system}) and (\ref{eq:example EPR/B}),
if we choose $c=2$ as a factual context in each of them, then the
respective F-CF subsystem will be 
\begin{equation}
\begin{array}{|c|c|c||c|}
\hline R_{2}^{1} & R_{3}^{1} &  & c=1\\
\hline R_{2}^{2} & R_{3}^{2} & R_{4}^{2} & \boxed{c=2}\vphantom{\begin{array}{c}
\\
\\
\end{array}}\\
\hline  & R_{3}^{3} &  & c=3\\
\hline R_{2}^{4} & R_{3}^{4} &  & c=4\\
\hline\hline q=2 & 3 & 4 & \mathcal{R}_{1}\left\llbracket c=2\right\rrbracket 
\\\hline \end{array}\:,\:\begin{array}{|c|c||c|}
\hline R_{2}^{1} &  & c=1\\
\hline R_{2}^{2} & R_{3}^{2} & \boxed{c=2}\vphantom{\begin{array}{c}
\\
\\
\end{array}}\\
\hline  & R_{3}^{3} & c=3\\
\hline\hline q=2 & 3 & \mathcal{R}_{2}\left\llbracket c=2\right\rrbracket 
\\\hline \end{array}\label{eq:F-CF example}
\end{equation}

\subsection{Counterfactual Definiteness}

We are ready now to rigorously formulate CFD in terms of the noncontextuality
of the F-CF subsystems of a system. 
\begin{defn}
\label{def:CFD}An s.c.c.\ system is said to have the property of
\emph{counterfactual definiteness} (CFD) if all its F-CF subsystems
are noncontextual.
\end{defn}
The justification of this definition lies in the intuition formalized
by the notion of an i.c.\ coupling. If the F-CF subsystems in (\ref{eq:F-CF example})
are noncontextual, then their i.c.\ couplings are, respectively,
\begin{equation}
\begin{array}{|c|c|c||c|}
\hline X & Y &  & c=1\\
\hline X & Y & Z & \boxed{c=2}\vphantom{\begin{array}{c}
\\
\\
\end{array}}\\
\hline  & Y &  & c=3\\
\hline X & Y &  & c=4\\
\hline\hline q=2 & 3 & 4 & \widetilde{\mathcal{R}}_{1}\left\llbracket c=2\right\rrbracket 
\\\hline \end{array}\:,\:\begin{array}{|c|c||c|}
\hline X' &  & c=1\\
\hline X' & Y' & \boxed{c=2}\vphantom{\begin{array}{c}
\\
\\
\end{array}}\\
\hline  & Y' & c=3\\
\hline\hline q=2 & 3 & \mathcal{\widetilde{R}}_{2}\left\llbracket c=2\right\rrbracket 
\\\hline \end{array}\label{eq:ic coupling}
\end{equation}
where $\left\{ X,Y,Z\right\} $ and $\left\{ X',Y'\right\} $ are
reduced couplings of the respective F-CF subsystems. It is \emph{as
if} all counterfactual variables were the same as the corresponding
factual ones. Moreover, in each of the counterfactual contexts these
representations of the two F-CF subsystems are compatible with the
overall joint distributions in this context: 
\begin{equation}
\begin{array}{|c|c|c|c|c||c|}
\hline R_{1}^{1} & X & Y &  &  & c=1\\
\hline  & X & Y & Z &  & \boxed{c=2}\vphantom{\begin{array}{c}
\\
\\
\end{array}}\\
\hline R_{1}^{3} &  & Y &  &  & c=3\\
\hline R_{1}^{4} & X & Y &  & R_{5}^{4} & c=4\\
\hline\hline q=1 & 2 & 3 & 4 & 5 & \mathrm{ext}.\widetilde{\mathcal{R}}_{1}\left\llbracket c=2\right\rrbracket 
\\\hline \end{array}\:,\:\begin{array}{|c|c|c|c||c|}
\hline R_{1}^{1} & X' &  &  & c=1\\
\hline  & X' & Y' &  & \boxed{c=2}\vphantom{\begin{array}{c}
\\
\\
\end{array}}\\
\hline  &  & Y' & R_{4}^{3} & c=3\\
\hline\hline q=1 & 2 & 3 & 4 & \mathrm{ext}.\mathcal{\widetilde{R}}_{2}\left\llbracket c=2\right\rrbracket 
\\\hline \end{array}
\end{equation}
where ``ext.'' abbreviates ``extended.'' This is a representation
(coupling) of the system (\ref{eq:example system}) with a factual
context $c=2$, in which any counterfactual variable sharing a content
with a factual one is identical with the latter. This is the intuitive
meaning of CFD.

\subsection{Universality of Counterfactual Definiteness}

It is easy to see by inspecting (\ref{eq:ic coupling}) that the noncontextuality
of this F-CF subsystem did not have to be assumed: e.g., if $\mathcal{R}_{1}$
in (\ref{eq:example system}) is s.c.c., the noncontextuality of $\mathcal{R}_{1}\left\llbracket c=2\right\rrbracket $
in (\ref{eq:F-CF example}) clearly holds by choosing the reduced
coupling $\left\{ X,Y,Z\right\} $ as $\left\{ R_{2}^{2},R_{3}^{2},R_{4}^{2}\right\} $.
The same holds for other three F-CF subsystems of (\ref{eq:example system}),
as well as for the four F-CF subsystems of (\ref{eq:example system}),
which means that the systems $\mathcal{R}_{1}$ and $\mathcal{R}_{2}$
both satisfiy CFD. 
\begin{thm}
\label{thm:F-CF-subsystem}Any s.c.c.\ system satisfies CFD. 
\end{thm}
\begin{proof}
Let $\mathcal{R}$ in (\ref{eq:system}) be a s.c.c.\ system. We
need to show that any of its F-CF subsystems is noncontextual. Let
$c=c_{0}$ to be a factual context. Then the jointly distributed set
of variables
\begin{equation}
R^{c_{0}}=\left\{ R_{q}^{c_{0}}:q\in Q,\,q\prec c_{0}\right\} 
\end{equation}
is a reduced coupling of $\mathcal{R}\left\llbracket c_{0}\right\rrbracket $.
Indeed, for any $c\in C$,
\begin{equation}
R^{c}=\left\{ R_{q}^{c}:q\in Q,\,q\prec c,c_{0}\right\} \overset{d}{=}\left\{ R_{q}^{c_{0}}:q\in Q,\,q\prec c_{0},c\right\} ,
\end{equation}
because the system is s.c.c.
\end{proof}
To emphasize once again the underlying intuition, it follows from
the theorem that any F-CF subsystem has an i.c.\ coupling
\begin{equation}
S^{*}\left\llbracket c_{0}\right\rrbracket =\left\{ S_{q}^{c}:S_{q}^{c}=R_{q}^{c_{0}},\,q\in Q,\,c\in C,\,q\prec c,c_{0}\right\} .
\end{equation}
In other words, for any counterfactual context $c\not=c_{0}$ and
any $q\prec c_{0},c$, the counterfactual variable $R_{q}^{c}$ can
be treated \emph{as if it were} $R_{q}^{c_{0}}$.

\section{Concluding Remarks and Possible Generalizations}

\subsection{Counterfactual Definiteness and Noncontextuality}

We have shown that compliance with CFD and noncontextuality of a system
of random variables are related but different concepts: CFD is always
satisfied as a consequence of non-disturbance (s.c.c.)\ property,
irrespective of whether the system is contextual. The mathematical
reason for this is that a counterfactual question creates a single
frame of reference, a factual context, which, together with the variables
in the counterfactual contexts to which the question pertains, forms
a noncontextual subsystem. The overall (non)contextuality is a property
without such a frame of reference, or one in which all different frames
of reference are reconciled. However, to achieve such a reconciliation,
or to establish that it is not possible, one has to use conceptual
means that cannot be presented in the form of counterfactual questions.

Interestingly, David Mermin, in his well-known paper \cite{Mermin1989},
makes a distinction between \emph{Strong Baseball Principle}, corresponding
to CFD, and \emph{Very Strong Baseball Principle}, corresponding to
overall noncontextuality. Except for terminological and expository
differences, the Strong Baseball Principle is introduced as in our
example (\ref{eq:Q1Q2}), by choosing a factual context and creating
two counterfactual ones following CFD. Then Mermin compares the two
counterfactual contexts to each other, and says: ``This last application
of the Strong Baseball Principle, by comparing hypothetical cases,
has a different character than the first two, which compare a hypothetical
case with the real one, and here it might more accurately be termed
the Very Strong Baseball Principle.'' Mermin proceeds, however, by
arguing that the latter should not be treated separately from the
Strong Baseball Principle, the argument being that comparing counterfactual
contexts is as ``reasonable'' and ``permissible'' as comparing
them with a factual context. One of the authors of the present paper
maintained the same position in Ref.~\cite{Dzh2019}, using similar
arguments. Our position here is that being equally reasonable and
permissible, CFD and noncontextuality are logically distinct principles,
of which only the latter may fail to hold in s.c.c.\ systems.

We find a similar situation in a 1981 paper by John Clauser and Abner
Shimony \cite{CS1981}. They analyze Henry Stapp's \cite{Stapp1971}
approach to the derivation of Bell's inequalities, according to which
they follow from certain four equations. These equations can be interpreted
as stating that CFD holds for the four F-CF subsystems of the EPR/Bohm
system, i.e. $\mathcal{R}_{2}$ in our example (\ref{eq:example EPR/B}).
Thus, Stapp's position is that whenever $\mathcal{R}_{2}$ is contextual,
it must contravene CFD. Clauser and Shimony point out that Stapp's
equations only apply to the pairs of contexts (F-CF pairs, in our
terminology), and that ``Stapp has not given a reason for demanding
the existence of a quadruple of possible worlds which mesh together
{[}these four pairs{]}.'' However, they seem to accept Stapp's response
to this objection, which we do not present here as, in our view, it
misses the point. Clauser and Shimony's objection is accepted as valid
by Bernard d’Espagnat \cite{dEspagnat1984}, but his position, in
contrast to ours, seems to be skeptical of the meaningfulness of CFD
altogether.

For completeness, we should mention that Robert Griffiths \cite{Griffiths2019,Griffith2012}
also argues that CFD is always satisfied in quantum mechanics. His
argumentation, however, is very different from ours. Moreover, unlike
in our paper, Griffiths considers CFD to be completely unrelated to
noncontextuality.

\subsection{Systems with Signaling/Disturbance}

CbD provides a generalization of the notion of (non)contextuality
to systems that are not necessarily s.c.c. In fact, the distributions
of the random variables in the system can be arbitrary. It is interesting
to see if CFD generalizes similarly, and if so, what the relations
are between generalized (non)contextuality and generalized CFD.
\begin{defn}
A system $\mathcal{R}$ in (\ref{eq:system}) is considered noncontextual
in CbD if it has a \emph{multimaximally connected} coupling, defined
as an indexed set of jointly distributed variables
\begin{equation}
S^{*}=\left\{ S_{q}^{c}:q\in Q,\,c\in C,\,q\prec c\right\} 
\end{equation}
such that, (1) for any $Q'\subseteq Q$ and any $c\in C$ with $q\prec c$
for all $q\in Q'$, the distribution of $\left\{ R_{q}^{c}:q\in Q'\right\} $
is the same as the distribution of $\left\{ S_{q}:q\in Q'\right\} $;
and (2) the probability of $S_{q}^{c}=S_{q}^{c'}$, for any $q\prec c,c'$,
is maximal possible.
\end{defn}
For s.c.c.\ systems, this definition specializes to Definition \ref{def:s.c.c. noncontextual},
with the maximal probability of $S_{q}^{c}=S_{q}^{c'}$ being 1 (because
of which they both can be renamed into $S_{q}$).

One can now apply Definition \ref{def:CFD} to an arbitrary system
$\mathcal{R}$: 
\begin{defn}
A system is said to satisfy generalized CFD (gCFD), if all of its
F-CF subsystems are noncontextual. 
\end{defn}
The intuitive meaning of gCFD is that the outcome of a counterfactual
measurement, had it been made, would have been the same as the factual
one with the highest possible probability (given the distributions
of the two measurements). 

We now have no analogue of Theorem \ref{thm:F-CF-subsystem}. Based
on the general properties of (non)contextuality, one can only say
that (A) if a system is noncontextual then so are all of its F-CF
subsystems (i.e., the system satisfies gCFD); and (B) if a system
does not satisfy gCFD, then it is contextual. Moreover, all s.c.c.\ systems
provide evidence that compliance with gCFD (in the form of CFD) does
not imply noncontextuality. All \emph{cyclic systems} \cite{Araujoetal.2013},
even if not s.c.c. \cite{DKC2020}, demonstrate the same. For instance,
the system $\mathcal{R}_{2}$ in our example (\ref{eq:example EPR/B})
 is contextual whenever
\begin{equation}
\begin{array}{r}
\max\left(\pm\left\langle R_{1}^{1}R_{2}^{1}\right\rangle \pm\left\langle R_{2}^{2}R_{3}^{2}\right\rangle \pm\left\langle R_{3}^{3}R_{4}^{3}\right\rangle \pm\left\langle R_{4}^{4}R_{1}^{4}\right\rangle \right)>2+\left|\left\langle R_{1}^{1}\right\rangle -\left\langle R_{1}^{4}\right\rangle \right|\\
+\left|\left\langle R_{2}^{2}\right\rangle -\left\langle R_{2}^{1}\right\rangle \right|+\left|\left\langle R_{3}^{3}\right\rangle -\left\langle R_{3}^{2}\right\rangle \right|+\left|\left\langle R_{4}^{4}\right\rangle -\left\langle R_{4}^{3}\right\rangle \right|,
\end{array}\label{eq:gBell}
\end{equation}
with the same meaning of the terms as in (\ref{eq:Bell}), which is
a special case of (\ref{eq:gBell}). At the same time, every F-CF
subsystem of this system is trivially noncontextual (because single-variable
rows cannot contribute to contextuality):
\begin{equation}
\begin{array}{cc}
\begin{array}{|c|c||c|}
\hline R_{1}^{1} & R_{2}^{1} & \boxed{c=1}\vphantom{\begin{array}{c}
\\
\\
\end{array}}\\
\hline  & R_{2}^{2} & c=2\\
\hline R_{1}^{4} &  & c=4\\
\hline\hline q=1 & 2 & \mathcal{R}_{2}\left\llbracket c=1\right\rrbracket 
\\\hline \end{array} & \begin{array}{|c|c||c|}
\hline R_{2}^{1} &  & c=1\\
\hline R_{2}^{2} & R_{3}^{2} & \boxed{c=2}\vphantom{\begin{array}{c}
\\
\\
\end{array}}\\
\hline  & R_{3}^{3} & c=3\\
\hline\hline q = 2 & 3 & \mathcal{R}_{2}\left\llbracket c=2\right\rrbracket 
\\\hline \end{array}\\
\\
\begin{array}{|c|c||c|}
\hline R_{3}^{2} &  & c=2\\
\hline R_{3}^{3} & R_{4}^{3} & \boxed{c=3}\vphantom{\begin{array}{c}
\\
\\
\end{array}}\\
\hline  & R_{4}^{4} & c=4\\
\hline\hline q = 3 & 4 & \mathcal{R}_{2}\left\llbracket c=3\right\rrbracket 
\\\hline \end{array} & \begin{array}{|c|c||c|}
\hline R_{1}^{1} &  & c=1\\
\hline  & R_{4}^{3} & c=3\\
\hline R_{1}^{4} & R_{4}^{4} & \boxed{c=4}\vphantom{\begin{array}{c}
\\
\\
\end{array}}\\
\hline\hline q=1 & 4 & \mathcal{R}_{2}\left\llbracket c=4\right\rrbracket 
\\\hline \end{array}
\end{array}
\end{equation}
More work is needed to find out if gCFD and more generally F-CF systems
may productively complement the notion of (non)contextuality in the
theory of systems of random variables. 

\vspace{6pt}

\paragraph*{Author Contributions:}

The initial idea belongs to J.V.K. The authors contributed equally
to all other aspects of the work. All authors have read and agreed to the published version of the manuscript.

\paragraph*{Funding:}
This research was partially supported by the Foundational Questions
Institute (grant FQxI-MGA-2201).

\paragraph*{Data Availability:}
No data is associated with the paper.

\paragraph*{Conflicts of Interest:}
The authors declare no conflict of interest. The sponsors had no
role in the design, execution, interpretation, or writing of the study.

\paragraph*{Abbreviations:}
The following abbreviations are used in this manuscript:\\

\noindent%
\begin{tabular}{@{}ll}
CbD & Contextuality-by-Default (approach to contextuality)\\
CFD & counterfactual definiteness\\
F-CF & factual-counterfactual (subsystem of random variables)\\
gCFD & generalized counterfactual definiteness\\
i.c. & identically connected (coupling)\\
MBR & ``magic box rule''\\
s.c.c. & strongly consistently connected (without disturbance, non-signaling)\\
\end{tabular}

\end{document}